\documentclass[%
    a4paper,
    12pt,
    DIV=12,
    BCOR=0mm,
    twoside=semi,
    headings=small
]{scrartcl}

\usepackage{scrpage2}

\addtokomafont{title}{\normalfont \bfseries}
\addtokomafont{section}{\normalfont \bfseries}
\addtokomafont{subsection}{\normalfont \bfseries}
\addtokomafont{subsubsection}{\normalfont \bfseries}
\addtokomafont{paragraph}{\normalfont \bfseries}
\newcommand{\quoteparagraph}[1]{\noindent{\normalsize \bfseries #1}\hspace{1em}}

\usepackage[utf8]{inputenc}           
\usepackage{amsmath, amsthm, amssymb} 
\usepackage[english]{babel}           
\usepackage{bibgerm}                  
\usepackage{graphicx}                 
\usepackage{hyperref}                 
\usepackage{booktabs}                 

\usepackage{etoolbox}
\usepackage{bm}
\usepackage{nicefrac}				
\usepackage{dsfont}				
\usepackage{acronym}  
\usepackage{environ}
\usepackage{tikz}

\newcommand{\R}{\ensuremath{\mathbb{R}}}

\newcommand{\N}{\ensuremath{\mathbb{N}}}

\newtheorem{theorem}{Theorem}[section]
\newtheorem{definition}[theorem]{Definition}
\newtheorem{assumption}[theorem]{Assumption}

\newtheorem{lemma}[theorem]{Lemma}
\newtheorem{remark}[theorem]{Remark}
\newtheorem{corollary}[theorem]{Corollary}

\makeatletter				

\NewEnviron{Lalign}{\tagsleft@true\begin{align}\BODY\end{align}}
\makeatother

\numberwithin{equation}{section}

\newcommand{\TWRtext}{Terminal Wealth Relative }
\newcommand{\system}[1][]{{\textit{(system #1)}}}
\newcommand{\that}{\hat{t}}	

\newcommand{\f}{\varphi}	
\newcommand{\vek}[1][\f]{\bm #1}												

\newcommand{\emptyvar}{\raisebox{-0.5ex}{\scalebox{1.6}{$\cdot$}}}

\newcommand{\vekt}[1][i]{\vek[t]_{#1\emptyvar}}
\newcommand{\quot}[1][i]{(\nicefrac{\vekt[#1]}{\vek[\that]})}
\newcommand{\msupp}{\mathfrak{G}}												

\newcommand{\msuppR}{\mathfrak{R}}												
\newcommand{\Hess}[1]{\operatorname{Hess}_{#1}}											
\newcommand{\sprod}[2]{\langle {#1},{#2}\rangle}										
\newcommand{\risk}{\operatorname{r}}

\newcommand{\acHPR}{Holding Period Return (HPR) }
\newcommand{\acTWR}{\TWRtext\ (TWR) }
\newcommand{\TWRo}{\text TWR }

\newcommand{\TWR}[1][\empty]{\ifthenelse{\equal{#1}{\empty}}{\operatorname{TWR}}{\operatorname{TWR}_{#1}}}	  
\newcommand{\HPR}[1][\empty]{\ifthenelse{\equal{#1}{\empty}}{\operatorname{HPR}}{\operatorname{HPR}_{#1}}}    

\usepackage{bookmark}
\hypersetup{
  pdftitle    = {Existence and Uniqueness for the Multivariate Discrete Terminal Wealth Relative},
  pdfauthor   = {Andreas Hermes and Stanislaus Maier-Paape},
  pdfkeywords = {},
  pdfborder={0 0 0.5}
}

\author{
  \normalsize \textsc{Andreas Hermes and Stanislaus Maier-Paape}\\[-0.2em]
    \small \textit{Institut f\"ur Mathematik, RWTH Aachen,}\\[-0.5em]
    \small \textit{Templergraben 55, D-52062 Aachen, Germany}\\[-0.5em]
    \small \href{mailto:ahermes@instmath.rwth-aachen.de}{ahermes@instmath.rwth-aachen.de}\\[-0.5em]
    \small \href{mailto:maier@instmath.rwth-aachen.de}{maier@instmath.rwth-aachen.de}
}
\date{
  \vspace{0.25em}
  \normalsize\today
  \vspace{-1cm}
}
\title{
  \vspace{-2cm}
  \Large Existence and Uniqueness for the Multivariate Discrete Terminal Wealth Relative
}

\clearscrheadfoot
\lehead{\pagemark}                   
\rohead{\pagemark}                   
\newif\ifpdfgraphics  
\pdfgraphicsfalse

\begin{document}

\maketitle

\vspace*{0.2cm}
\begin{quote}
  \small
  \quoteparagraph{Abstract}
  In this paper the multivariate fractional trading ansatz of money management from Vince \cite{vince:pmf90} is discussed.
  In particular, we prove existence and uniqueness of an ``optimal $f$'' of the respective optimization problem under
  reasonable assumptions on the trade return matrix. This result generalizes a similar result for the univariate fractional
  trading ansatz. Furthermore, our result guarantees that the multivariate optimal $f$ solutions can always be found
  numerically by steepest ascent methods. \\

  \quoteparagraph{Keywords} fractional trading, optimal f, multivariate discrete terminal wealth relative,
                            risk and money management, portfolio theory
  \end{quote}


\vspace*{-0.1cm}
            \section{Introduction}   \label{sec:introduction}

 Risk and money management for investment issues has always been at the heart of\\ finance. Going back to the
 1950s, Markowitz \cite{markowitz:pfs1991} invented the ``modern portfolio theory'',\\ where the additive
 expectation of a portfolio of different investments was maximized subject to a given risk expressed by volatility
 of the portfolio.

 When the returns of the portfolio are no longer calculated additive, but multiplicative in order to respect
 the needs of compound interest, the resulting optimization problem is known as ``fixed fractional trading''.
 In fixed fractional trading strategies an investor always wants to risk a fixed percentage of his current
 capital for future investments given some distribution of historic trades of his trading strategy.

 A first example of factional trading was established in the 1950s by Kelly \cite{kelly:nii} who found a criterion
 for an asymptotically optimal investment strategy for one investment instrument. Similarly, Vince in the 1990s
 (see \cite{vince:pmf90} and \cite{vince:mmm92}) used the fractional trading ansatz to optimize his position sizing.
 Although at first glance these two methods look quite different, they are in fact closely related as could be shown in
 \cite{maier:raft2016}. However, only recently in \cite{vince:rpm09}, Vince extended the fractional trading ansatz
 to portfolios of different investment instruments. The situation with $M$ investment instruments (systems) and $N$
 coincident realizations of absolute returns of these $M$ systems results in a trade return matrix $T$ described in
 detail in \eqref{eq:returnMatrix}. Given this trade return matrix, the ``\TWRtext\!\!'' \text{(TWR)} can be constructed
 (see \eqref{eq:HPRvec}) measuring the multiplicative gain of a portfolio resulting from a fixed vector
 $\vek=(\f_1,\dots,\f_M)$ of fractional investments into the $M$ systems. In order to find an optimal investment
 among all fractions $\vek$ the \TWRo has to be maximized
 \begin{align}\label{eq:TWRmax} 
   \underset{\vek\in\msupp}{\text{maximize}}\quad\TWR(\vek) ,
 \end{align}
 where $\msupp$ is the definition set of the \TWRo (see Definition~\ref{def:msupp} and \eqref{prob:discm}).

 Whereas in \cite{vince:rpm09}, Vince only stated this optimization problem and illustrated it with examples, in
 Section~\ref{sec:3} we give as our main result the necessary analysis. In particular, we investigate the definition
 set $\msupp$ of the \TWRo and fix reasonable assumptions (Assumption~\ref{discm}) under which \eqref{eq:TWRmax} has a
 unique solution. This unique solution may lie in $\overset{\circ}{\msupp}$ or on $\partial\msupp$ as different examples in
 Section \ref{sec:4} show. Our result extend the results of Maier--Paape \cite{maier:eto2013}, Zhu \cite{zhu:mais2007}
 ($M=1$\  case only) and parts of the PhD of Hermes \cite{hermes:mft2016} on the discrete multivariate \text{TWR}.
 One of the main ingredients to show the uniqueness of the maximum of \eqref{eq:TWRmax} is the concavity of the function
 $\left[\text{TWR}(\cdot)\right]^{1/N}$ (see Lemma~\ref{lem:mconcavity}). Uniqueness and concavity furthermore guarantee that
 the solution of \eqref{eq:TWRmax} can always be found numerically by simply following steepest ascent.

 Before we start our analysis, some more remarks on related papers are in order. \\
 In \cite{maier:optf2015} Maier--Paape showed that the fractional trading ansatz on one investment instrument leads to
 tremendous drawdowns, but that effect can be reduced largely when several stochastic independent trading systems are
 used coincidentally. Under which conditions this diversification effect works out in the here considered multivariate
 \TWRo situation is still an open question. Furthermore, several papers investigated risk measures in the context of
  fractional trading with one investment instrument
 ($M=1$; see \cite{prado:orb2013}, \cite{maier:eto2013},  \cite{maier:raft2016}  and \cite{vince:ips2013}).
 Related investigations for the multivariate \TWRo  using the drawdown can be found in Vince \cite{vince:rpm09}. \\

 In the following sections we now analyse the multivariate case of a discrete \TWRtext\!\!.
 That means we consider multiple investment strategies where every strategy generates multiple trading returns.
 As noted before this situation can be seen as a portfolio approach of a discrete \TWRtext \ (cf. \cite{vince:rpm09}).
 For example one could consider an investment strategy applied to several assets, the strategy producing trading returns on
 each asset. But in an even broader sense, one could also consider several distinct investment strategies applied to several
 distinct assets or even classes of assets.


      \section{Definition of a \TWRtext}  \label{sec:2}

\vspace{-4mm}
  The subject of consideration in this paper is the multivariate case of the discrete \TWRtext for several trading systems analogous to the definition of Ralph Vince
  in \cite{vince:rpm09}. For $1\leq k \leq M,\,M\in\N,$ we denote the $k$-th trading system by \system[k]. A trading system is an investment strategy applied to a
  financial instrument. Each system generates periodic trade returns, e.g. monthly, daily or the like. The absolute trade return of the $i$-th period of the $k$-th
  system is denoted by $t_{i,k}$, $1\leq i \leq N, 1\leq k\leq M$. Thus we have the joint return matrix

 \begin{align*}
 \begin{tabular}{c|cccc}
    period	 & \system[1] & \system[2]	& $\cdots$	& \system[M]  \\
  \hline
    $1$		 & $t_{1,1}$  & $t_{1,2}$	& $\cdots$	& $t_{1,M}$   \\
    $2$		 & $t_{2,1}$  & $t_{2,2}$	& $\cdots$	& $t_{2,M}$   \\
    $\vdots$ & $\vdots$	  & $\vdots$	& $\ddots$	& $\vdots$    \\
   $N$		 & $t_{N,1}$  & $t_{N,2}$	& $\cdots$	& $t_{N,M}$   \\
 \end{tabular}
 \end{align*}

 and define
 \begin{align}\label{eq:returnMatrix}\index{historical trading returns} 
    T:=\bigg(t_{i,k}\bigg)_{\substack{1\leq i \leq N\\ 1\leq k \leq M}}\in\R^{N\times M}.
 \end{align}

 Just as in the univariate case (cf. \cite{maier:eto2013} or \cite{vince:pmf90}), we assume that each system produced at least one loss
 within the $N$ periods. That means
 \begin{align}\label{eq:loss_hist} 
   \forall\, k\in\{1,\dots,M\}\,\exists\, i_0=i_0(k)\in\{1,\dots,N\} \text{ such that } t_{i_0,k}<0
 \end{align}

\noindent
 Thus we can define the biggest loss of each system as
 $$\that_k:=\max_{1\leq i \leq N}\{|t_{i,k}|\mid t_{i,k}<0\}>0,\quad1\leq k \leq M.$$ \index{biggest loss}

\noindent
 For better readability, we define the rows  of the given return matrix, i.e. the return of the $i$-th period, as
 $$\vekt:=(t_{i,1},\dots,t_{i,M})\in\R^{1\times M}$$
 and the vector of all biggest losses as
 $$\vek[\that]:=(\that_1,\dots,\that_M)\in\R^{1\times M}.$$
 Having the biggest loses at hand, it is possible to ``normalize'' the $k$--th column of $T$ by $1/\that_k$
 such that each system has a maximal loss of $-1$. Using the componentwise quotient, the normalized trade matrix return then
 has the rows
 $$\quot := \left(\frac{t_{i,1}}{\that_1},\dots,\frac{t_{i,M}}{\that_M}\right) \in \R^{1\times M},\quad 1\leq i\leq N\,.$$


 For $\vek:=(\f_1,\dots,\f_M)^\top$, $\f_k\in[0,1]$, we define the \acHPR of the $i$-th period as \index{Holding Period Return}
 \begin{align}\label{eq:HPRvec} 
   \HPR[i](\vek):=1+\sum\limits_{k=1}^M\f_k\frac{t_{i,k}}{\that_k}=1+\sprod{\quot^\top}{\vek}_{\R^M},
 \end{align}

 where $\sprod{\emptyvar}{\emptyvar}_{\R^M}$ denotes the standard scalar product on $\R^M$. To shorten the notation, the marking of the vector space $\R^M$ at
 the scalar product is omitted, if the dimension of the vectors is clear. Similar to the univariate case, the gain (or loss) in each system is scaled by its
 biggest loss. Therefore the $\HPR$ represents the gain (loss) of one period, when investing a fraction of $\nicefrac{\f_k}{\that_k}$ of the capital in
 \system[k] for all $1\leq k\leq M$, thus risking a maximal loss of $\f_k$ in the $k$-th trading system.\\

 The \acTWR as the gain (or loss) after the given $N$ periods, when the fraction $\f_k$ is invested in \system[k] over all periods, is then given as
 \begin{align}\label{eq:TWRvec}\index{Terminal Wealth Relative!discrete} 
   \begin{aligned}
      \TWR_N(\vek): &= \prod\limits_{i=1}^N\HPR[i](\vek)  \\
	    	        &= \prod\limits_{i=1}^N\left( 1+\sum\limits_{k=1}^M\f_k\frac{t_{i,k}}{\that_k}\right)=\prod\limits_{i=1}^N \left(1+\sprod{\quot^\top}{\vek}\right).
   \end{aligned}
 \end{align}

 Note that in the $M=1$--dimensional case a risk of a full loss of our capital corresponds to a fraction of $\f=1\in\R$. Here in the multivariate case we have a loss
 of $100\%$ of our capital every time there exists an $i_0\in\{1,\dots,N\}$ such that $\HPR[i_0](\vek)=0$. That is for example if we risk a maximal loss of
 $\f_{k_0}=1$ in the $k_0$-th trading system (for some $k_0\in\{1,\dots,M\}$) and simultaneously letting $\f_k=0$ for all other $k\in\{1,\dots,M\}$.
 However these {\textit degenerate} vectors of fractions are not the only examples that produce a \acTWR of zero. Since we would like to risk at most $100\%$ of our
 capital (which is quite a meaningful limitation), we restrict $\TWR_N:\msupp\to\R$ to the domain $\msupp$ given by the following definition:
\begin{definition}\label{def:msupp} 
   A vector of fractions $\vek\in\R_{\geq0}^M$ is called {\textit admissible} if $\vek\in\msupp$ holds, where
   \begin{align*}\index{admissible vector of fractions}
      \msupp:&=\{\vek\in\R_{\geq0}^M\mid\HPR[i](\vek)\geq0,\,\forall\,1\leq i\leq N\}\\
	         &=\{\vek\in\R_{\geq0}^M\mid\sprod{\quot^\top}{\vek}\geq-1,\,\forall\,1\leq i\leq N\}.
   \end{align*}
   Furthermore we define
   \begin{align*}
     \msuppR&:=\{\vek\in\msupp\mid\,\exists\,1\leq i_0\leq N\text{ s.t. }\HPR[i_0](\vek)=0\}.
   \end{align*}
\end{definition}

 With this definition we now have a risk of $100\%$ for each vector of fractions $\vek\in\msuppR$ and a risk of less than $100\%$ for each vector of fractions
 $\vek\in\msupp\setminus\msuppR$.
 Since
 $$\HPR[i](\vek[0])=1\quad\text{for all }1\leq i \leq N$$
 we can find an $\varepsilon>0$ such that
 $$\Lambda_\varepsilon:=\{\vek\in\R_{\geq0}^M\mid \|\vek\|\leq\varepsilon\}\subset\msupp,$$
 and thus in particular $\msupp\neq\emptyset$ holds. $\|\emptyvar\|=\sqrt{\sprod{\emptyvar}{\emptyvar}}$ denotes the Euclidean norm on $\R^M$.

 Observe that the $i$-th period results in a loss if $\HPR[i](\vek)<1$, that means $\sprod{\quot^\top}{\vek}=\HPR[i](\vek)-1<0$.
 Hence the biggest loss over all periods for an investment with a given vector of fractions $\vek\in\msupp$ is
 \begin{align}\label{eq:riskvec}\index{biggest loss} 
   \risk(\vek):=\max\left\{-\min\limits_{1\leq i\leq N}\{\sprod{\quot^\top}{\vek}\},0\right\}.
 \end{align}
 Consequently, we have a biggest loss of
 \begin{alignat*}{2}
   \risk(\vek)&=1\quad &&\forall\,\vek\in\msuppR
   \intertext{and}
   \risk(\vek)&\in[0,1)\quad\quad &&\forall \vek\in\msupp\setminus\msuppR.
 \end{alignat*}
 Note that for $\vek\in\msupp$ we do not have an a priori bound for the fractions $\f_k$, $k=1,\dots,M$. Thus it may happen that there are $\vek\in\msupp\setminus\msuppR$
 with $\f_k>1$ for some (or even for all) $k\in\{1,\dots,M\}$, or at least $\sum\limits_{k=1}^M \f_k > 1$, indicating a risk of more than $100\%$ for the individual trading
 systems, but the combined risk of all trading systems $\risk(\vek)$ can still be less than $100\%$. So the individual risks can potentially be eliminated to some extent
 through diversification. As a drawback of this favorable effect the optimization in the multivariate case may result in vectors of fractions $\vek\in\msupp$ that require
 a high capitalization of the individual trading systems. Thus we assume leveraged financial instruments and ignore margin calls or other regulatory issues.

 Before we continue with the \TWRo analysis, let us state a first auxiliary lemma for $\msupp$.

\begin{lemma}\label{lem:convexity} 
  The set $\msupp$ in Definition~\ref{def:msupp} is convex, as is $\msupp\setminus\msuppR$.
\end{lemma}

\begin{proof}
	All the conditions $\f_k \ge 0,\, k=1,\dots,M$ and
	\begin{align*}
		\HPR[i](\vek)\geq0\quad\Leftrightarrow\quad\sprod{\quot^\top}{\vek}\geq-1,\quad i=1,\dots,N
	\end{align*}
	define half spaces (which are convex). Since $\msupp$ is the intersection of a finite set of half spaces, it is itself convex.
	
	A similar reasoning yields that $\msupp\setminus\msuppR$ is convex, too.
\end{proof}


\section{Optimal Fraction of the Discrete \TWRtext}   \label{sec:3}
\allowdisplaybreaks[1]

 If we develop this line of thought a little further a necessary condition for the return matrix $T$ for the optimization of the \TWRtext gets clear:

\begin{lemma}\index{risk free investment}
  Assume there is a vector $\vek_0\in \Lambda_\varepsilon$ with $\risk(\vek_0)=0$ then
  $$\{s\cdot\vek_0\mid s\in\R_{\geq0}\}\subset\msupp\setminus\msuppR.$$
  If in addition there is an $1\leq i_0\leq N$ such that $\HPR_{i_0}(\vek_0)>1$ then
  $$\TWR[N](s\cdot\vek_0)\xrightarrow[s\to\infty]{}\infty.$$
\end{lemma}
\begin{proof}
  If
  $$\risk(\vek_0)=\max\left\{-\min\limits_{1\leq i\leq N}\{\sprod{\quot^\top}{\vek_0}\},0\right\}=0,$$
  it follows that
  \begin{align}\label{eq:positive_HPR}
     \HPR[i](\vek_0)\geq1\quad\text{ for all }1\leq i \leq N.
  \end{align}
  For arbitrary $s\in\R_{\geq0}$ the function
  $$s\mapsto\HPR[i](s\vek_0)=1+\sprod{\quot^\top}{s\vek_0}=1+s\underbrace{\sprod{\quot^\top}{\vek_0}}_{\geq0}\geq1$$
  is monotonically increasing in $s$ for all $i=1,\dots,N$ and by that we have
  $$s\vek_0\in\msupp\setminus\msuppR.$$
  Moreover, if there is an $i_0$ with $\HPR_{i_0}(\vek_0)>1$ then
  $$\HPR_{i_0}(s\vek_0)\xrightarrow[s\to\infty]{}\infty$$
  and by that
  $$\TWR[N](s\cdot\vek_0)\xrightarrow[s\to\infty]{}\infty.$$
\end{proof}

 An investment where the holding period returns are greater than or equal to $1$ for all periods denotes a ``risk free'' investment ($r(\vek)=0$) and considering the possibility 
 of an unbounded leverage, it is clear that the overall profit can be maximized by investing an infinite quantity. Assuming arbitrage free investment instruments, any risk free 
 investment can only be of short duration, hence by increasing $N\in\N$ the condition $\HPR[i](\vek_0)\geq1$ will eventually burst, cf. (\ref{eq:positive_HPR}). 
 Thus, when optimizing the \TWRtext, we are interested in settings that fulfill the following assumption
 \begin{align*}
    \forall \, \vek\in \partial B_\varepsilon(0)\cap\Lambda_\varepsilon\,\,\exists\,i_0=i_0(\vek) \text{ such that } \sprod{\quot[i_0]^\top}{\vek}<0,
 \end{align*}
 always yielding $r(\vek)>0$.

 With that at hand, we can formulate the optimization problem for the multivariate discrete \TWRtext
 \begin{align}\label{prob:discm}\index{optimization problem!discrete}
   	\underset{\vek\in\msupp}{\text{maximize}}\quad\TWR[N](\vek)
  \end{align}
 and analyze the existence and uniqueness of an optimal vector of fractions for the problem under the assumption
 \begin{assumption}\label{discm}
	We assume that each of the trading systems in (\ref{eq:returnMatrix}) produced at least one loss (cf. (\ref{eq:loss_hist})) and furthermore
	\begin{Lalign}
		&\begin{aligned}
			&\forall \, \vek\in \partial B_\varepsilon(0)\cap\Lambda_\varepsilon\,\,\exists\,i_0=i_0(\vek)\in\{1,\dots,N\}\\[-1mm]
			& \text{such that }\sprod{\quot[i_0]^\top}{\vek}<0 \quad\quad\quad (\text{\bfseries no risk free investment})
		\end{aligned}\tag{a}{\label{asm1:discm}}\\[4mm]
		&\begin{aligned}
			&\frac{1}{N}\sum\limits_{i=1}^N t_{i,k}>0\quad\forall\,k=1,\dots,M   \quad(\text{\bfseries each trading system is profitable})
		\end{aligned}\tag{b}{\label{asm2:discm}}\\[4mm]
       &\ker(T)=\{\vek[0]\}\tag{c}{\label{asm3:discm}}   \quad\quad\quad\quad\quad  \quad \quad  \quad     (\text{\bfseries linear independent trading systems})
	\end{Lalign}
 \end{assumption}

 Assumption~\ref{discm}(\ref{asm1:discm}) ensures that, no matter how we allocate our portfolio (i.e. no matter what direction $\vek\in\msupp$ we choose), there is always at least one 
 period that realizes a loss, i.e. there exists an $i_0$ with $\HPR[i_0](\vek)<1$. Or in other words, not only are the investment systems all fraught with risk (cf. (\ref{eq:loss_hist})), 
 but there is also no possible risk free allocation of the systems.

 The matrix $T$ from (\ref{eq:returnMatrix}) can be viewed as a linear mapping
 \begin{align*}
    T:\R^M\to\R^M,
 \end{align*}
 ``$\ker(T)$'' denotes the kernel of the matrix $T$ in Assumption~\ref{discm}(\ref{asm3:discm}). Thus this assumption is the linear independence of the trading systems, i.e. the 
 linear independence of the columns
 $$\vek[t]_{\emptyvar k}\in\R^N,\quad k\in\{1,\dots,M\}$$
 of the matrix $T$. Hence with Assumption~\ref{discm}(\ref{asm3:discm}) it is not possible that there exists an
 $1\leq k_0\leq M$ and a $\vek[\psi]\in\R^M\setminus\{\vek[0]\}$ such that
 $$(-\psi_{k_0})\begin{pmatrix}t_{1,k_0}\\\vdots\\t_{N,k_0}\end{pmatrix}=\sum\limits_{\substack{k=1\\k\neq k_0}}^M\psi_k\begin{pmatrix}t_{1,k}\\\vdots\\t_{N,k}\end{pmatrix},$$
 which would make \system[$k_0$] obsolete. So Assumption~\ref{discm}(\ref{asm3:discm}) is no actual restriction of the optimization problem.

 Now we point out a first property of the \TWRtext.
\begin{lemma}\label{lem:mbounded1}
  Let the return matrix $T\in\R^{N\times M}$ (as in (\ref{eq:returnMatrix})) satisfy Assumption~\ref{discm}(\ref{asm1:discm}) then, for all $\vek\in\msupp\setminus\{\vek[0]\}$, 
  there exists an $s_0=s_0(\vek)>0$ such that $\TWR[N](s_0\vek)=0$. In fact $s_0\vek\in\msuppR$.
\end{lemma}
\begin{proof}
  For some arbitrary $\vek\in\msupp\setminus\{0\}$ we have $\frac{\varepsilon}{\|\vek\|}\cdot\vek\in \partial B_\varepsilon(0)\cap\Lambda_\varepsilon$. 
  Then Assumption~\ref{discm}(\ref{asm1:discm}) yields the existence of an
  $i_0\in\{1,\dots,N\}$ with $\sprod{\quot[i_0]^\top}{\vek}<0$. With
  $$j_0:=\operatorname*{argmin}\limits_{1\leq i\leq N}\{\sprod{\quot^\top}{\vek}\}\in\{1,\dots,N\}$$
  and
  $$s_0:=-\frac{1}{\sprod{\quot[j_0]^\top}{\vek}}>0$$
  we get that
  $$\HPR[j_0](s_0\vek)=1+\sprod{\quot[j_0]^\top}{s_0\vek}=1+s_0\sprod{\quot[j_0]^\top}{\vek}=0$$
  and $\HPR[i](s_0\vek)\geq0$ for all $i\neq j_0$. Hence $\TWR[N](s_0\vek)=0$ and clearly $s_o\vek\in\msuppR$ (cf.~Definition~\ref{def:msupp}).
\end{proof}

\noindent 
 Furthermore the following holds.
\begin{lemma}\label{lem:mboundedsupport}
  Let the return matrix $T\in\R^{N\times M}$ (as in (\ref{eq:returnMatrix})) satisfy Assumption~\ref{discm}(\ref{asm1:discm}) then the set $\msupp$ is compact.
\end{lemma}
\begin{proof}
  For all $\vek\in\partial B_\varepsilon(0)\cap\Lambda_\varepsilon$ Assumption~\ref{discm}(\ref{asm1:discm}) yields an $i_0(\vek)\in\{1,\dots,N\}$ such that
  $\sprod{\quot[i_0]^\top}{\vek}<0$. With that we define
  $$m:\partial B_\varepsilon(0)\cap\Lambda_\varepsilon\to\R,\vek\mapsto m(\vek):=\min\limits_{1\leq i\leq N}\{\sprod{\quot^\top}{\vek}\}<0.$$
  This function is continuous on the compact support $\partial B_\varepsilon(0)\cap\Lambda_\varepsilon$. Thus the maximum exists
  $$M:=\max\limits_{\vek\in \partial B_\varepsilon(0)\cap\Lambda_\varepsilon}m(\vek)<0.$$
  Consequently the function
  $$g:\partial B_\varepsilon(0)\cap\Lambda_\varepsilon\to\R_{\geq0}^M,\vek\mapsto\frac{1}{|m(\vek)|}\cdot\vek$$
  is well defined and continuous.
  Since for all $\vek\in\partial B_\varepsilon(0)\cap\Lambda_\varepsilon$
  $$\sprod{\quot^\top}{\frac{1}{|m(\vek)|}\vek}=\frac{\sprod{\quot^\top}{\vek}}{|\min\limits_{1\leq i\leq N}\{\sprod{\quot^\top}{\vek}\}|}\geq-1\quad\forall\,1\leq i \leq N$$
  with equality for at least one index $\tilde i_0\in\{1,\dots,N\}$, we have
  $$\HPR_i\left(\frac{1}{|m(\vek)|}\vek\right)\geq0\quad\forall 1\leq i\leq N$$
  and
  $$\HPR_{\tilde i_0}\left(\frac{1}{|m(\vek)|}\vek\right)=0,$$
  hence
  $$\frac{1}{|m(\vek)|}\vek\in\msuppR.$$
  Altogether we see that
  \begin{align*}
     g\left(\partial B_\varepsilon(0)\cap\Lambda_\varepsilon\right)&=\left\lbrace\frac{1}{|m(\vek)|}\cdot\vek\mid \vek\in\partial B_\varepsilon(0)\cap\Lambda_\varepsilon\right\rbrace=\msuppR,\\
  \end{align*}
  thus the set $\msuppR$ is bounded and connected as image of the compact set $\partial B_\varepsilon \cap \Lambda_\varepsilon$ under the continuous function $g$ and by that  
  the set $\msupp$ is compact.
\end{proof}

\noindent 
 Now we take a closer look at the third assumption for the optimization problem.
\begin{lemma}\label{lem:mconcavity}	
  Let the return matrix $T\in\R^{N\times M}$ (as in (\ref{eq:returnMatrix})) satisfy Assumption~\ref{discm}(\ref{asm3:discm}) then $\TWR[N]^{\nicefrac{1}{N}}$ is concave on $\msupp\setminus\msuppR$. Moreover if there is a $\vek_0\in\msupp\setminus\msuppR$ with $\nabla\TWR[N](\vek)=\vek[0]$, then
  $\TWR[N]^{\nicefrac{1}{N}}$ is even strictly concave in $\vek_0$.
\end{lemma}
\begin{proof}
  For $\vek\in\msupp\setminus\msuppR$ the gradient of $\TWR[N]^{\nicefrac{1}{N}}$ is given by the column vector
  \begin{align}
		&\nabla\TWR[N]^{\nicefrac{1}{N}}(\vek)\notag\\
			&=\TWR[N]^{\nicefrac{1}{N}}(\vek)\cdot\frac{1}{N}\sum\limits_{i=1}^N\frac{1}{1+\sum\limits_{k=1}^M\f_k\frac{t_{i,k}}{\that_k}}\cdot
					\begin{pmatrix}\nicefrac{t_{i,1}}{\that_1}\\\nicefrac{t_{i,2}}{\that_2}\\\vdots\\\nicefrac{t_{i,M}}{\that_M}\end{pmatrix}\notag\\
				&=\TWR[N]^{\nicefrac{1}{N}}(\vek)\cdot\frac{1}{N}\sum\limits_{i=1}^N\frac{1}{1+\sprod{\quot^\top}{\vek}}\cdot\quot^\top\in\R^{M\times1},\label{eq:grad_TWRdiscm}
  \end{align}
  where $\TWR[N]^{\nicefrac{1}{N}}(\vek)>0$. The Hessian-matrix is then given by
  \begin{align*}
		&\Hess{\TWR[N]^{\nicefrac{1}{N}}}(\vek)\\
		      &=\nabla\left[\left(\nabla\TWR[N]^{\nicefrac{1}{N}}(\vek)\right)^\top\right]\\
		      &=\nabla \left[\TWR[N]^{\nicefrac{1}{N}}(\vek)\cdot\frac{1}{N}\sum\limits_{i=1}^N\frac{1}{1+\sprod{\quot^\top}{\vek}}\quot\right]\\
		      &=\nabla\TWR[N]^{\nicefrac{1}{N}}(\vek)\cdot\frac{1}{N}\sum\limits_{i=1}^N\frac{1}{1+\sprod{\quot^\top}{\vek}}\quot\\
			&\omit\hfill $\displaystyle+\TWR[N]^{\nicefrac{1}{N}}(\vek)\frac{1}{N}\sum\limits_{i=1}^N\left(-\frac{1}{(1+\sprod{\quot^\top}{\vek})^2}\quot^\top\cdot\quot\right)$\\
					&=\TWR[N]^{\nicefrac{1}{N}}(\vek)\Bigg[\underbrace{\frac{1}{N^2}\sum\limits_{i=1}^N \vek[y]_i^\top\sum\limits_{i=1}^N \vek[y]_i
			-\frac{1}{N}\sum\limits_{i=1}^N \vek[y]_i^\top \vek[y]_i}_{=:-\nicefrac{1}{N}\cdot B(\vek)\in\R^{M\times M}}\Bigg]
  \end{align*}
  where $\vek[y]_i:=\frac{1}{1+\sprod{\quot^\top}{\vek}}\quot\in\R^{1\times M}$ is a row vector. The matrix $B(\vek)$ can be rearranged as
  \begin{align*}
		B(\vek)
			&=\sum\limits_{i=1}^N \vek[y]_i^\top \vek[y]_i-\frac{1}{N}\left(\sum\limits_{i=1}^N \vek[y]_i^\top\right)\left(\sum\limits_{i=1}^N \vek[y]_i\right)\\
			&=\sum\limits_{i=1}^N \vek[y]_i^\top \vek[y]_i-\frac{1}{N}\left[\sum\limits_{i=1}^N \vek[y]_i^\top\left(\sum\limits_{u=1}^N \vek[y]_u\right)\right]
			      -\frac{1}{N}\left[\sum\limits_{i=1}^N\left(\sum\limits_{v=1}^N \vek[y]_v^\top \right)\vek[y]_i\right]\\
			    &\quad+\frac{1}{N^2}\left(\sum\limits_{i=1}^N 1\right)\left(\sum\limits_{v=1}^N \vek[y]_v^\top\right)\left(\sum\limits_{u=1}^N \vek[y]_u\right)\\
			&=\sum\limits_{i=1}^N\left[\vek[y]_i^\top \vek[y]_i-\vek[y]_i^\top\frac{1}{N}\left(\sum\limits_{u=1}^N \vek[y]_u\right)-\frac{1}{N}\left(\sum\limits_{v=1}^N 
                  \vek[y]_v^\top\right)\vek[y]_i\right.\\
			    &\quad\left.+\frac{1}{N^2}\left(\sum\limits_{v=1}^N \vek[y]_v^\top\right)\left(\sum\limits_{u=1}^N \vek[y]_u\right)\right]\\
			&=\sum\limits_{i=1}^N\left[\vek[y]_i^\top\left(\vek[y]_i-\frac{1}{N}\sum\limits_{u=1}^N \vek[y]_u\right)
			      -\frac{1}{N}\left(\sum\limits_{v=1}^N \vek[y]_v^\top\right)\left(\vek[y]_i-\frac{1}{N}\sum\limits_{u=1}^N \vek[y]_u\right)\right]\\
			&=\sum\limits_{i=1}^N\Bigg[\left(\vek[y]_i^\top-\frac{1}{N}\sum\limits_{v=1}^N \vek[y]_v^\top\right)
			    \Bigg(\underbrace{\vek[y]_i-\frac{1}{N}\sum\limits_{u=1}^N \vek[y]_u}_{:=\vek[w]_i\in\R^{1\times M}}\Bigg)\Bigg]\\
			&=\sum\limits_{i=1}^N \vek[w]_i^\top \vek[w]_i.
  \end{align*}
  Since the matrices $\vek[w]_i^\top \vek[w]_i$ are positive semi-definite for all $i=1,\dots,N$, the same holds for $B(\vek)$ and therefore $\TWR[N]^{\nicefrac{1}{N}}$ is concave.
  Furthermore if there is a $\vek_0\in\msupp\setminus\msuppR$ with
  \begin{alignat*}{2}
			&&\nabla\TWR[N](\vek_0)&=0\\
	 &\stackrel{\TWR[N](\vek_0)>0}{\Leftrightarrow}&\quad\sum\limits_{i=1}^N\frac{1}{1+\sprod{\quot^\top}{\vek_0}}\quot&=0\\
	 &\stackrel{\phantom{\TWR[N](\f_0)>0}}{\Leftrightarrow}&\quad\sum\limits_{i=1}^N \vek[y]_i &=0,
  \end{alignat*}
  where $\vek[y]_i=\vek[y]_i(\vek_0)$,	the matrix $B(\vek_0)$ further reduces to
  \begin{align*}
	  B(\vek_0)=\sum\limits_{i=1}^N \vek[y]_i^\top \vek[y]_i.
  \end{align*}
  If $B(\vek_0)$ is not strictly positive definite there is a $\vek[\psi]=(\psi_1,\dots,\psi_M)^\top\in\R^M\setminus\{\vek[0]\}$ such that
  \begin{align*}
	  0=\vek[\psi]^\top B(\vek_0)\vek[\psi]=\sum\limits_{i=1}^N\vek[\psi]^\top \vek[y]_i^\top \vek[y]_i \vek[\psi]=\sum\limits_{i=1}^N\underbrace{\sprod{\vek[y]_i^\top}{\vek[\psi]}^2}_{\geq0}
  \end{align*}
  and we get that
  \begin{alignat*}{2}
	 &\phantom{\Leftrightarrow}\quad&\sprod{\vek[y]_i^\top}{\vek[\psi]}=\frac{1}{1+\sprod{\quot^\top}{\vek_0}}\sprod{\quot^\top}{\vek[\psi]}&=0\quad\forall\,1\leq i \leq N\\
	 &\Leftrightarrow&\sprod{\quot^\top}{\vek[\psi]}&=0\quad\forall\,1\leq i \leq N,
  \end{alignat*}
  yielding a non trivial element in $\ker(T)$ and thus contradicting Assumption~\ref{discm}(\ref{asm3:discm}). Hence matrix $B(\vek_0)$ is strictly positive definite and $\TWR[N]^{\nicefrac{1}{N}}$ is strictly concave in $\vek_0$.
\end{proof}

\noindent 
 With this at hand we can state an existence and uniqueness result for the multivariate optimization problem.
\begin{theorem}\label{thm:optf_discm}\index{existence and uniqueness result!discrete}\sloppy 
  ({\bfseries optimal $f$ existence}) Given a return matrix $T=\bigg(t_{i,k}\bigg)_{\substack{1\leq i \leq N\\ 1\leq k \leq M}}$ as in (\ref{eq:returnMatrix}) that fulfills 
  Assumption~\ref{discm}, then there exists a solution $\vek^{opt}_N\in\msupp$ of the optimization problem~(\ref{prob:discm})
  \begin{align}\label{eq:thm36} 
	  \underset{\vek\in\msupp}{\text{maximize}}\quad\TWR[N](\vek).
  \end{align}
  Furthermore one of the following statements holds:
  \begin{enumerate}
		\item  $\vek^{opt}_N$ is unique, or
		\item $\vek^{opt}_N\in\partial \msupp$.
  \end{enumerate}\fussy
  For both cases $\vek^{opt}_N\neq0$, $\vek^{opt}_N\notin\msuppR$ and ${\TWR[N](\vek^{opt}_N)>1}$ hold true.
\end{theorem}
\begin{proof}
  We show existence and partly uniqueness of a maximum of the $N$-th root of $\TWR[N]$, yielding existence and partly uniqueness of a solution $\vek^{opt}_N$ 
  of (\ref{eq:thm36}) with the claimed properties.\\
  With Lemma~\ref{lem:convexity} and Lemma~\ref{lem:mboundedsupport}, the support $\msupp$ of the \TWRtext is convex and compact. Hence the
  continuous function $\TWR[N]^{\nicefrac{1}{N}}$ attains its maximum on $\msupp$.
\sloppy
  For $\vek=\vek[0]$ we get from (\ref{eq:grad_TWRdiscm})
  \begin{align*}
	  \nabla\TWR[N]^{\nicefrac{1}{N}}(\vek[0])&=\underbrace{\TWR[N]^{\nicefrac{1}{N}}(\vek[0])}_{=1}\cdot\frac{1}{N}\sum\limits_{i=1}^N\quot^\top,
  \end{align*}
  which is a vector whose components are strictly positive due to Assumption~\ref{discm}(\ref{asm2:discm}). Therefore $\vek[0]\in\msupp$ is not a maximum of 
  $\TWR[N]^{\nicefrac{1}{N}}$	and a global maximum reaches a value greater than
  $$\TWR[N]^{\nicefrac{1}{N}}(\vek[0])=1.$$
  Since for all $\vek\in\msuppR$
  $$\TWR[N]^{\nicefrac{1}{N}}(\vek)=0$$
  holds, a maximum can not be attained in $\msuppR$ either.
	
  Now if there is a maximum on $\partial\msupp$, assertion (b) holds together with the claimed properties. Alternatively, a maximum $\vek_0$ is attained in the interior 
  $\mathring{\msupp}$. In this case, Lemma~\ref{lem:mconcavity} yields the strict concavity of $\TWR[N]^{\nicefrac{1}{N}}$ at $\vek_0$. Suppose there is another maximum $\vek^{\ast}\in\msupp\setminus\msuppR$ then the straight line connecting both maxima
  \begin{align*}
	 L:=\{t\cdot\vek_0+(1-t)\cdot\vek^{\ast}\mid \,t\in[0,1]\}
  \end{align*}
  is fully contained in the convex set $\msupp\setminus\msuppR$ (cf. Lemma~\ref{lem:convexity}).  Because of the concavity of $\TWR[N]^{\nicefrac{1}{N}}$ all points of $L$ have 
  to be maxima, which is a contradiction to the strict concavity of $\TWR[N]^{\nicefrac{1}{N}}$  in $\vek_0$. Thus the maximum is unique and assertion~(a) holds together with 
  the claimed properties.
\fussy
\end{proof}

 In the remaining of this section, we will further discuss case (b) in Theorem~\ref{thm:optf_discm}. We aim to show that the maximum $\vek^{opt}_N\in\partial\msupp$ is unique either, 
 but we proof this using a completely different idea. In order to lay the grounds for this, first, we give a lemma:

\begin{lemma}\label{lem:subdimT} 
  If $T\in\R^{N\times M}$ from (\ref{eq:returnMatrix}) is a return map satisfying Assumption~\ref{discm} and if $M\geq2$, then each return map $\tilde{T}\in\R^{N\times(M-1)}$, 
  which results from $T$ after eliminating one of its columns, is also a return map satisfying Assumption~\ref{discm}.
\end{lemma}
\begin{proof}
  Since each of the $M$ trading systems of the return matrix $T\in\R^{N\times M}$ has a biggest loss $\that_k$, $1\leq k \leq M$, the same holds for the $(M-1)$ trading systems 
  of the reduced matrix $\tilde{T}\in\R^{N\times(M-1)}$.

  For $\tilde{T}$, Assumption~\ref{discm}~(\ref{asm2:discm}) and (\ref{asm3:discm}) follow straight from the respective properties of the matrix $T$.

  Now let, without loss of generality, $\tilde{T}$ be the matrix that results from $T$ by eliminating the last column, i.e. the $M$-th trading system is omitted. 
  Let $\vekt^{(M-1)}\in\R^{M-1}$, $i=1,\dots,N$, denote the rows of $\tilde{T}$ and $\vek[\that]^{(M-1)}\in\R^{M-1}$ the vector of biggest losses of $\tilde{T}$. 
  Then for Assumption~\ref{discm}~(\ref{asm1:discm}) we have to show that
  \begin{align*}
	 \forall\, \vek^{(M-1)}\in\partial B_{\varepsilon}^{(M-1)}(0)\cap\Lambda_{\varepsilon}^{(M-1)}\,\,\exists\,i_0=i_0(\vek^{(M-1)})\in\{1,\dots,N\}, \notag
  \end{align*}
  such that	
  \begin{align}\label{eq:asm1ex3} 
	\sprod{(\nicefrac{\vekt^{(M-1)}}{\vek[\that]^{(M-1)}})^\top}{\vek^{(M-1)}}<0.
  \end{align}
  Using Assumption~\ref{discm}~(\ref{asm1:discm}) for matrix $T$ and
  \begin{align*}
  	\vek^{M}:=\begin{pmatrix}\f_1^{(M-1)}\\\vdots\\\f_{M-1}^{(M-1)}\\0\end{pmatrix}\in\partial B_{\varepsilon}^{(M)}(0)\cap\Lambda_{\varepsilon}^{(M)},
  \end{align*}
  the inequality
  \begin{align*}
 	\sprod{\quot^\top}{\vek^{(M)}}<0,
  \end{align*}
  holds true. Thus (\ref{eq:asm1ex3}) holds likewise.
\end{proof}

 Having this at hand, we can now extend Theorem~\ref{thm:optf_discm}.
\begin{corollary}\label{cor:uniquenessboundary} 
   ({\bfseries optimal $f$ uniqueness}) In the situation of Theorem~\ref{thm:optf_discm} the uniqueness also holds for case (b), i.e. a maximum $\vek_N^{opt}\in\partial\msupp$ is also 
   a unique maximum of $\TWR[N](\vek)$ in $\msupp$.
\end{corollary}
\begin{proof}
  Assume that the optimal solution $\vek_0:=\vek_N^{opt}\in\partial \msupp$ is not unique, then there exists an additional optimal solution $\vek^\ast\in\partial\msupp$ with $\vek^\ast\neq\vek_0$. Since $\msupp\setminus\msuppR$ is convex (c.f. Lemma~\ref{lem:convexity}), the line connecting both solutions
  \begin{align*}
		L:=\{t\cdot\vek_0+(1-t)\cdot\vek^{\ast}\mid \,t\in[0,1]\}
  \end{align*}
  is fully contained in $\msupp\setminus\msuppR$. Because of the concavity of $\TWR[N]^{\nicefrac{1}{N}}$ on $\msupp\setminus\msuppR$ (c.f. Lemma~\ref{lem:mconcavity}), all 
  points on $L$ are optimal solutions. Therefore $L$ must be a subset of $\partial\msupp\setminus\msuppR$, since we have seen that an optimal solution in the interior 
  $\mathring{\msupp}$ would be unique. Hence, there is (at least) one $k_0\in\{1,\dots,M\}$ such that, for all investment vectors in $L$, the trading system \system[$k_0$] 
  is not invested . I.e. the $k_0$-th component of $\vek_0$, $\vek^\ast$ and all vectors in $L$ is zero.

  Without loss of generality, let $k_0=M$. Then
  \begin{align*}
     \vek_0=\begin{pmatrix}\f_1\\\vdots\\\f_{M-1}\\0\end{pmatrix}\neq\begin{pmatrix}\f^\ast_1\\\vdots\\\f^\ast_{M-1}\\0\end{pmatrix}=\vek^\ast
  \end{align*}
  are two optimal solutions for
  \begin{align*}
	\TWR[N](\vek)\stackrel{\mbox{\Large !}}{=}\max
  \end{align*}
  But with that, the $(M-1)$-dimensional investment vectors $\vek^{(M-1)}_0:=(\tilde\f_1,\dots,\tilde\f_{M-1})^\top$ and $\vek^{\ast,(M-1)}:=(\f^\ast_1,\dots,\f^\ast_{M-1})^\top$
  are two distinct optimal solutions for
  \begin{align*}
     \TWR[N]^{(M-1)}(\begin{pmatrix}\f_1\\\vdots\\\f_{M-1}\end{pmatrix}):=\prod\limits_{i=1}^{N}\left(1+\sum\limits_{k=1}^{M}\f_k\frac{t_{i,k}}{\that}\right)
     \stackrel{\mbox{\Large !}}{=}\max\,.
   \end{align*}

  With Lemma~\ref{lem:subdimT} the return map $\tilde T\in\R^{N\times(M-1)}$, which results from $T$ after eliminating the $M$-th column (i.e. \system[$M$]) satisfies 
  Assumption~\ref{discm}. Applying Theorem~\ref{thm:optf_discm} to the sub-dimensional optimization problem, yields that $\vek^{(M-1)}_0$ and $\vek^{\ast,(M-1)}$ again 
  lie at the boundary of the admissible set of investment vectors $\msupp^{(M-1)}\subset\R^{M-1}$.

  Hence, we have two distinct optimal solutions on the boundary $\partial\msupp^{(M-1)}$ for the optimization problem with $(M-1)$ investment systems. By induction this 
  reasoning leads to the existence of two distinct optimal solutions for an optimization problem with just one single trading system. But for that type of problem, we 
  already know that the solution is unique (see for example \cite{maier:eto2013}), which causes a contradiction to our assumption.
  Thus, also for case (b) we have the uniqueness of the solution $\vek^{opt}_N\in\partial \msupp$.
\end{proof}

\begin{remark}\label{rem:uniqueness} 
Note that Assumption \ref{discm}(\ref{asm3:discm}) is necessary for uniqueness. To give a counterexample imagine a return matrix $T$ with two equal columns,
meaning the same trading system is used twice. Let $\varphi^{opt}$ be the optimal f for this one dimensional trading system. Then it is easy to see that $(\varphi^{opt},0)$,
$(0,\varphi^{opt})$ and the straight line connecting these two points yield TWR optimal solutions for the return matrix $T$.

\end{remark}


\section{Example}   \label{sec:4}

 As an example we fix the joint return matrix $T:=\left(t_{i,k}\right)_{\substack{1\leq i\leq 6\\1\leq k\leq 4 }}$ for $M=4$ trading systems and the returns from $N=6$ periods 
 given through the following table.
 \begin{align}\label{eq:return_ex1}
 \begin{tabular}{c|rrrr}
    period	& \system[1]	            & \system[2]	            & \system[3]     & \system[4]     \\   \hline
    $1$ 	& $2\quad\quad$			    & $1\quad\quad$			    & $-1\quad\quad$ & $1\quad\quad$  \\
    $2$		& $2\quad\quad$			    & $-\tfrac{1}{2}\quad\quad$	& $2\quad\quad$	 & $-1\quad\quad$ \\
    $3$ 	& $-\tfrac{1}{2}\quad\quad$	& $1\quad\quad$			    & $-1\quad\quad$ & $2\quad\quad$  \\
    $4$		& $1\quad\quad$			    & $2\quad\quad$			    & $2\quad\quad$	 & $-1\quad\quad$ \\
    $5$ 	& $-\tfrac{1}{2}\quad\quad$	& $-\tfrac{1}{2}\quad\quad$	& $2\quad\quad$	 & $1\quad\quad$  \\
    $6$		& $-1\quad\quad$			& $-1\quad\quad$			& $-1\quad\quad$ & $-1\quad\quad$ \\
 \end{tabular}
 \end{align}
 Obviously every system produced at least one loss within the $6$ periods, thus the vector $\vek[\that]=(\that_1,\that_2,\that_3,\that_4)^\top$ with
 \begin{align*}
    \that_k=\max\limits_{1\leq i\leq 6}\{|t_{i,k}|\mid t_{i,k}<0\}=1,\quad k=1,\dots,4,
 \end{align*}
 is well-defined. For $\vek\in\msupp\setminus\msuppR$ the $\TWR[6]$ takes the form
 \begin{align*}\index{Terminal Wealth Relative!discrete}
    \TWR[6](\vek)=&(1+2\f_1+\f_2-\f_3+\f_4)(1+2\f_1-\tfrac{1}{2}\f_2+2\f_3-\f_4)\\
         &(1-\tfrac{1}{2}\f_1+\f_2-1\f_3+2\f_4)(1+\f_1+2\f_2+2\f_3-\f_4)\\
         &(1-\tfrac{1}{2}\f_1-\tfrac{1}{2}\f_2+2\f_3+1\f_4)(1-\f_1-\f_2-\f_3-\f_4),
 \end{align*}
 where the set of admissible vectors is given by
 \begin{align*}\index{admissible vector of fractions}
      \msupp&=\{\vek\in\R_{\geq0}^4\mid\sprod{\quot^\top}{\vek}\geq-1,\,\forall\,1\leq i\leq 6\}\\
			&=\{\vek\in\R^4_{\geq0}\mid\sprod{\quot[6]^\top}{\vek}=\min\limits_{i=1,\dots,6}\sprod{\quot^\top}{\vek}\geq-1\}\\
            &=\{\vek\in[0,1]^4\mid\f_1+\f_2+\f_3+\f_4\leq1\}.
 \end{align*}
 Since for all $\vek\in\msupp$
 \begin{align*}
		\sprod{\quot^\top}{\vek}\geq\sprod{\quot[6]}{\vek}\geq-1\quad\forall\,i=1,\dots,6
 \end{align*}
 we have
 \begin{align*}
   \sprod{\quot^\top}{\vek}=-1\text{ for some }i\in\{1,\dots,6\}\quad\Rightarrow\quad\sprod{\quot[6]^\top}{\vek}=-1.
 \end{align*}
 Accordingly we get
 \begin{align*}
    \msuppR&=\{\vek\in\msupp\mid\,\exists\,1\leq i_0\leq 6\text{ s.t. }\sprod{\quot^\top}{\vek}=-1\}\\
	       &=\{\vek\in[0,1]^4\mid\f_1+\f_2+\f_3+\f_4=1\}.
 \end{align*}
 When examining the $6$-th row $\vekt[6]=(-1,-1,-1,-1)$ of the matrix $T$ we observe that Assumption \ref{discm}(\ref{asm1:discm}) is fulfilled with $i_0=6$. To see that let, for 
 some $\varepsilon>0$, $\vek\in\partial B_\varepsilon\cap\Lambda_\varepsilon$, then
 \begin{align*}
   \sprod{\quot[6]^\top}{\vek}=-\f_1-\f_2-\f_3-\f_4<0.
 \end{align*}
 For Assumption \ref{discm}(\ref{asm2:discm}) one can easily check that all four systems are ``profitable'', since the mean values of all four columns in (\ref{eq:return_ex1}) are 
 strictly positive.
 Lastly, for Assumption \ref{discm}(\ref{asm3:discm}) we check that the rows of matrix $T$ are linearly independent
 \begin{align*}
   \det\begin{vmatrix}\vekt[1]\\ \vekt[2]\\ \vekt[3]\\ \vekt[4]\end{vmatrix}=\det\begin{vmatrix}2&1&-1&1\\2&-\tfrac12&2&-1\\-\tfrac12&1&-1&2\\1&2&2&-1\end{vmatrix}=22.75\neq0.
 \end{align*}
 Thus Theorem \ref{thm:optf_discm} yields the existence and uniqueness of an optimal investment fraction $\vek^{opt}_6\in\msupp$ with $\vek^{opt}_6\neq0$, $\vek^{opt}_6\notin\msuppR$ 
 and $\TWR[6](\vek^{opt}_6)>1$, which can numerically be computed
 \begin{align*}
   \vek^{opt}_6\approx\begin{pmatrix}0.2362\\0.0570\\0.1685\\0.1012\end{pmatrix}.
 \end{align*}

 In the above example, a crucial point is that there is one row in the return matrix where the $k$-th entry is the biggest loss of \system[k], $k=1,\dots,6$. Such a row in the return matrix implies, that all trading systems realized their biggest loss simultaneously, which can be seen as a strong evidence against a sufficient diversification of the systems. Hence we analyze Assumption \ref{discm}(\ref{asm1:discm}) a little closer to see what happens if this is not the case.

 With the help of Assumption~\ref{discm}(\ref{asm1:discm}), for all $\vek\in \partial B_\varepsilon(0)\cap\Lambda_\varepsilon$, there is a row of the return matrix $\vekt[i_0]$, 
 $i_0\in\{1,\dots,N\}$ such that $\sprod{\quot[i_0]^\top}{\vek}<0$. The sets
 \begin{align*}
	\{\vek\in\R^M\mid\sprod{\quot^\top}{\vek}=0\},\quad i=1,\dots,N
 \end{align*}
 describe the hyperplanes generated by the normal direction $\quot^\top\in\R^M$, $i=1,\dots,N$. Thus each $\vek$ from the set $ \partial B_\varepsilon(0)\cap\Lambda_\varepsilon$ 
 has to be an element of one of the half spaces
 \begin{align*}
	H_i:=\{\vek\in\R^M\mid\sprod{\quot^\top}{\vek}\leq0\}, \quad i=1,\dots,N.
 \end{align*}
 In other words the set  $\partial B_\varepsilon(0)\cap\Lambda_\varepsilon$ has to be a subset of a union of half spaces
 \begin{align*}
	 \left(\partial B_\varepsilon(0)\cap\Lambda_\varepsilon\right)\subset\bigcup_{i=1}^N H_i.
 \end{align*}
 If there exists an index $i_0$ such that $t_{i_0,k}=-\that_k$ for all $1\leq k\leq M$, then the normal direction of the corresponding hyperplane is
 \begin{align*}
	\quot[i_0]^\top=\begin{pmatrix}-1\\-1\\\vdots\\-1\end{pmatrix}\in\R^M,
 \end{align*}
 hence
 \begin{align*}
	\left(\partial B_\varepsilon(0)\cap\Lambda_\varepsilon\right)\subset\R^M_{\geq0}\subset H_{i_0}
 \end{align*}
 and therefore Assumption~\ref{discm}(\ref{asm1:discm}) is fulfilled.
 Figure \ref{fig:hyperplanes2} shows a hyperplane for $M=2$ and a row of the return matrix where all entries are the biggest losses, that means the normal direction 
 of this hyperplane is the vector
\begin{align*}
	\nicefrac{\begin{pmatrix}-\that_1\\ -\that_2\end{pmatrix}}{\begin{pmatrix}\that_1\\\that_2\end{pmatrix}}=\begin{pmatrix}-1\\-1\end{pmatrix}.
\end{align*}

\noindent
\begin{figure}[htbp]
		\centering
		\begin{tikzpicture}
				\node[anchor=south west,inner sep=0] at (0,0) {\includegraphics[width=.49\textwidth]{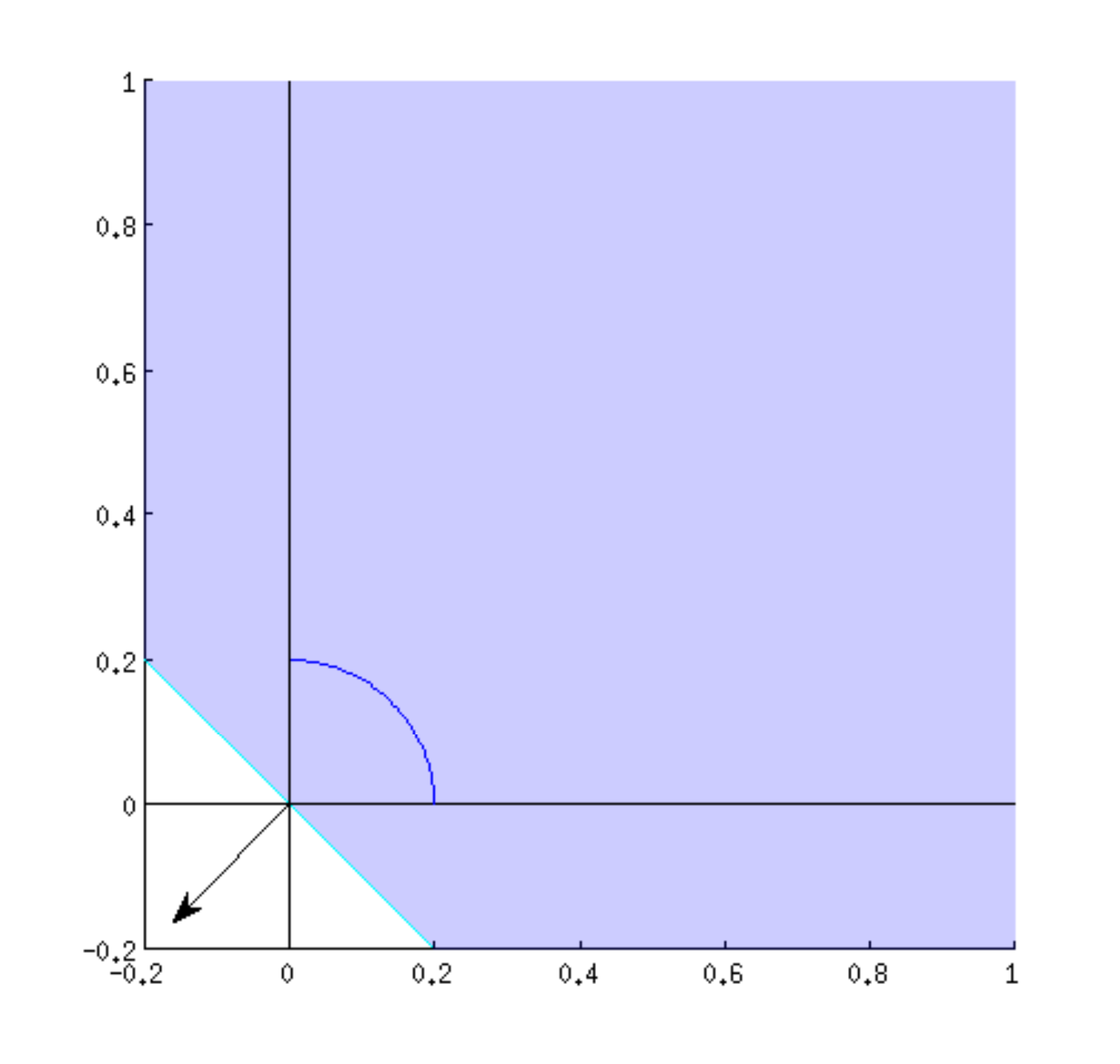}};
				\node at (1.5,6.1) {$\scriptstyle\f_2$};\node at (6.3,1.4) {$\scriptstyle\f_1$};
		\end{tikzpicture}
		\caption{Hyperplane for a return vector consisting of ``biggest losses''}\label{fig:hyperplanes2}
\end{figure}

 However, it is not necessary for Assumption \ref{discm}(\ref{asm1:discm}) that the set $\partial B_\varepsilon(0)\cap\Lambda_\varepsilon$ is covered by just one hyperplane.
 Again for $M=2$ an illustration of possible hyperplanes can be seen in Figure \ref{fig:hyperplanes1}. The figure on the left shows a case where Assumption~\ref{discm}(\ref{asm1:discm}) 
 is violated and the figure on the right a case where it is satisfied.

\begin{figure}[ht]\centering
		\begin{tikzpicture}
				\node[anchor=south west,inner sep=0] at (0,0) {\includegraphics[width=.49\textwidth]{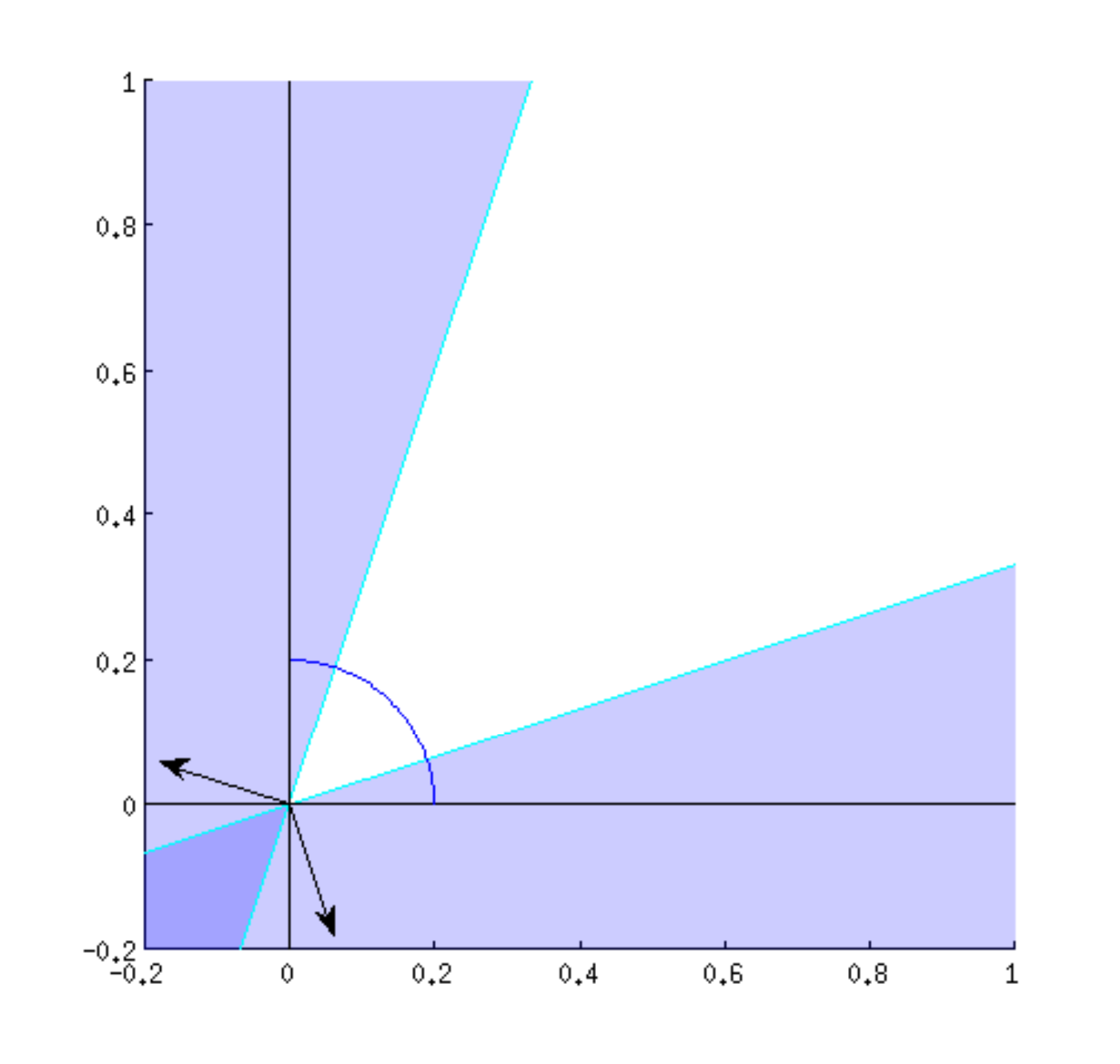}};
				\node at (1.5,6.1) {$\scriptstyle\f_2$};\node at (6.3,1.4) {$\scriptstyle\f_1$};
		\end{tikzpicture}
				\begin{tikzpicture}
				\node[anchor=south west,inner sep=0] at (0,0) {\includegraphics[width=.49\textwidth]{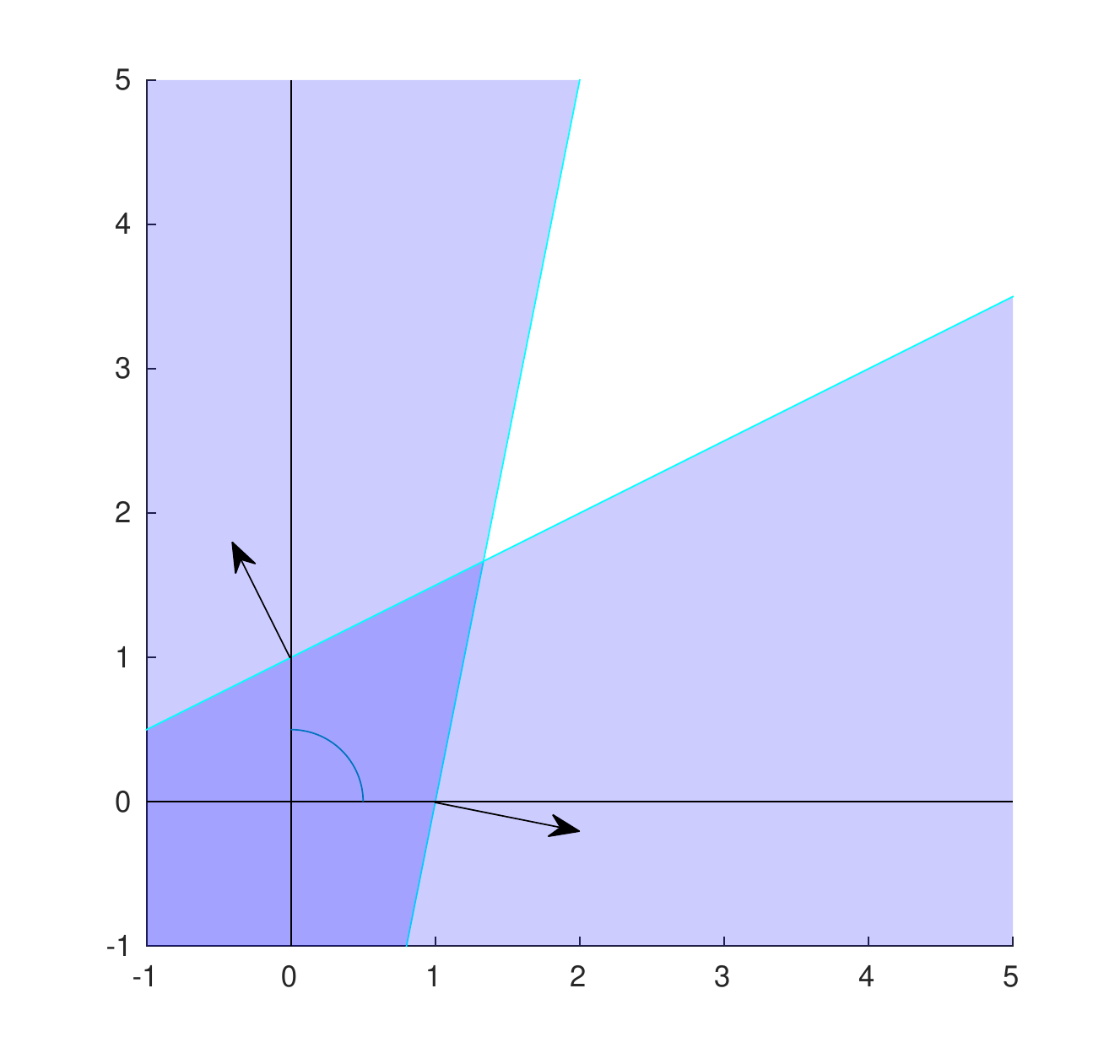}};
				\node at (1.5,6.1) {$\scriptstyle\f_2$};\node at (6.3,1.4) {$\scriptstyle\f_1$};
		\end{tikzpicture}
 \caption{Two hyperplanes and the set $	\partial B_\varepsilon(0)\cap\Lambda_\varepsilon$}\label{fig:hyperplanes1}
\end{figure}

 For the next example we fix the return matrix $T$ as
 \begin{align}\label{eq:return_ex2}
	T:=\frac{1}{5}\begin{pmatrix}-3 & 3\\ 9 & 12\\ 6 & -3\\ -6 & \nicefrac{3}{2}\\ 3 & -\nicefrac{15}{2}\\\end{pmatrix},
 \end{align}
 with $N=5$ and $M=2$. Thus the biggest losses of the two systems are
 $$\that_1=\frac{6}{5}\quad\text{and}\quad\that_2=\frac{3}{2}.$$
 To determine the set of admissible investments (and to check  Assumption \ref{discm}) we examine the vectors $\quot$ for $i=1,\dots,5$
 \begin{align}\label{eq:normedreturnmatrix}
    A:=\begin{pmatrix}-\nicefrac{1}{2}&\nicefrac{2}{5}\\ \nicefrac{3}{2}&\nicefrac{8}{5}\\ 1&-\nicefrac{2}{5}\\ -1&\nicefrac{1}{5}\\ \nicefrac{1}{2}&-1\\ \end{pmatrix}
 \end{align}
 and solve the linear equations
 \begin{align}\label{eq:admissible}
				\sprod{\quot^\top}{\vek}=-1,\quad i=1,\dots,5.
 \end{align}
 The solutions for $i=1,\dots,5$ are shown in Figure \ref{fig:admissible}.

\noindent
\begin{figure}[ht]\centering
		\begin{tikzpicture}
				\node[anchor=south west,inner sep=0] at (0,0) {\includegraphics[width=.8\textwidth]{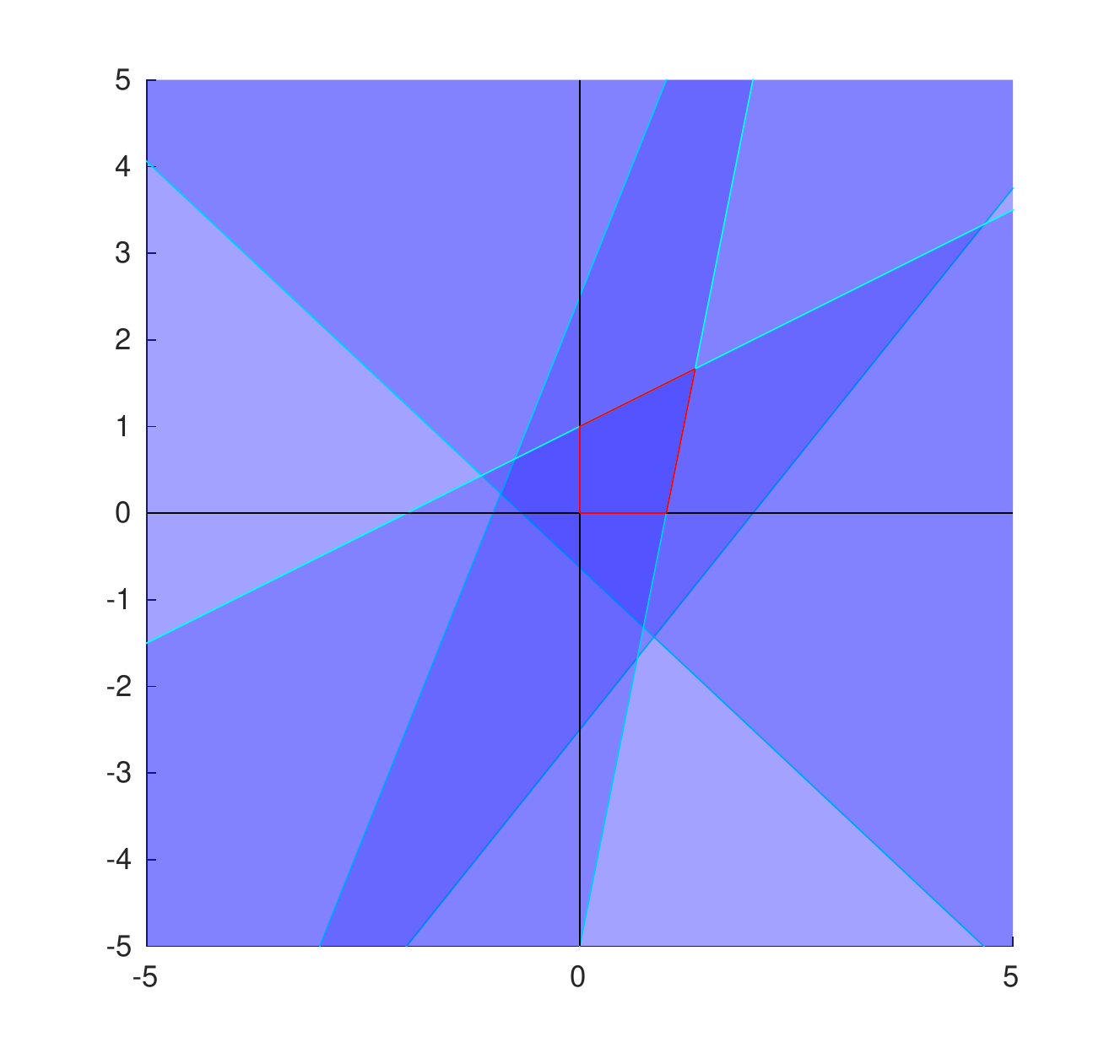}};
				\node at (5.4,9.8) {$\f_2$};\node at (10.2,5.1) {$\f_1$};
		\end{tikzpicture}
 \caption{Solutions of the linear equations from (\ref{eq:admissible})}\label{fig:admissible}
\end{figure}

 Each solution corresponds to a ``cyan'' line. The area where the inequality $	\sprod{\quot^\top}{\vek}\geq-1$ holds for some $i\in\{1,\dots,5\}$ is shaded in ``light blue''. 
 The set where the inequalities hold for all $i=1,\dots,5$ is the section where all shaded areas overlap, thus the ``dark blue'' section. Therefore the set of admissible 
 investments is given by
 \begin{align*}\index{admissible vector of fractions}
	\msupp&=\{\vek\in\R_{\geq0}^2\mid \sprod{\quot^\top}{\vek}\geq-1,\,\forall\,1\leq i\leq 5\}\\
									&=\{\vek\in\R^2_{\geq0}\mid \f_2\leq 1+\tfrac12\f_1\text{ and } \f_1\leq 1+\tfrac15\f_2\},
 \end{align*}
 with
 \begin{align*}
    \msuppR&=\{\vek\in\msupp\mid\,\exists\,1\leq i_0\leq 5\text{ s.t. }\sprod{\quot^\top}{\vek}=-1\}\\
									  &=\{\vek\in\R^2_{\geq0}\mid \f_2= 1+\tfrac12\f_1\text{ or } \f_1= 1+\tfrac15\f_2\}.
 \end{align*}
 Assumption \ref{discm} is fulfilled, since
 \begin{enumerate}
	\item the half spaces for rows $4$ and $5$ of the return matrix cover the whole set $\R^2_{\geq0}$ (cf. Figure~\ref{fig:hyperplanes1} b),
	\item $\frac{1}{5}\sum\limits_{i=1}^5 t_{i,1}=\frac{9}{5}>0$ and $\frac{1}{5}\sum\limits_{i=1}^5 t_{i,2}=\frac{6}{5}>0$ and
	\item obviously, the columns of the return matrix are linearly independent.
 \end{enumerate}
 A plot of the \TWRtext for the return matrix $T$ from (\ref{eq:return_ex2}) can be seen in Figure \ref{fig:TWR_ex2_full} and  \ref{fig:TWR_ex2_above} with a maximum at
 \begin{align}\label{eq:optfdiscmex}
	\vek^{opt}_5\approx\begin{pmatrix}0.4109\\0.3425\end{pmatrix}.
 \end{align}
 
\begin{figure}[ht]\centering
	\begin{tikzpicture}
					\node[anchor=south west,inner sep=0] at (0,0) {\includegraphics[width=.8\textwidth]{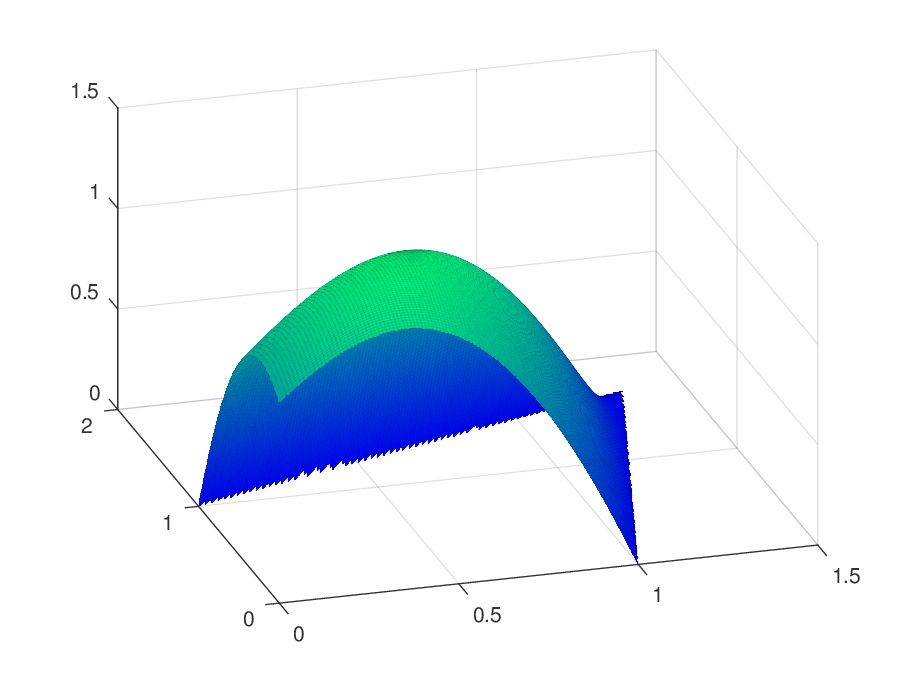}};
			\node at (0.5,3.1) {$\f_2$};\node at (10.5,0.8) {$\f_1$};\node at (1.2,7.6) {$\TWR[5](\f_1,\f_2)$};
	\end{tikzpicture}
 \caption{The \TWRtext for $T$ from (\ref{eq:return_ex2})}\label{fig:TWR_ex2_full}
\end{figure}
\begin{figure}[ht]\centering
	\begin{tikzpicture}
				\node[anchor=south west,inner sep=0] at (0,0) {\includegraphics[width=.8\textwidth]{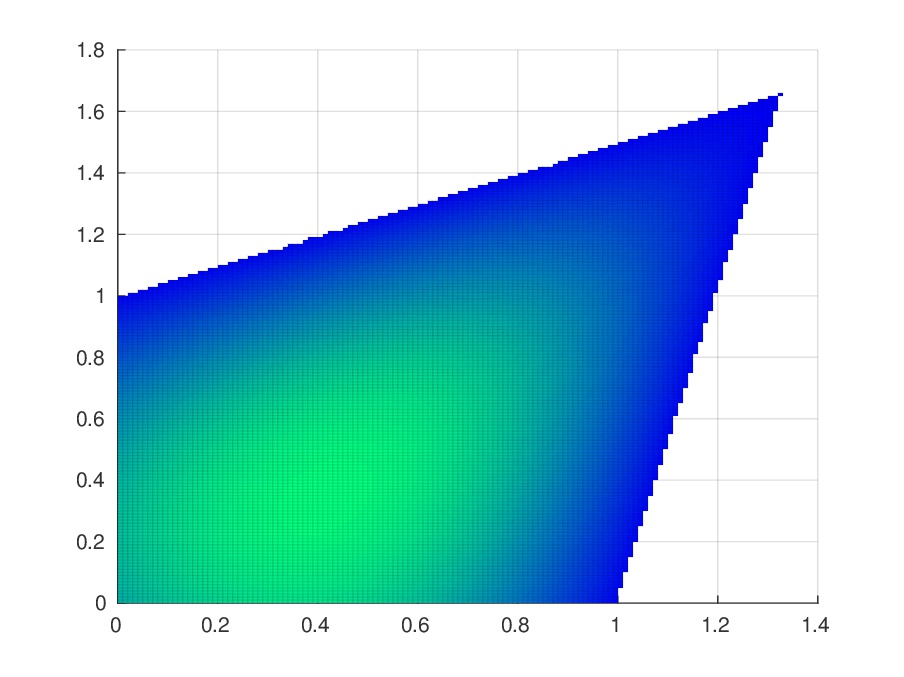}};
				\node at (0.6,7.8) {$\f_2$};\node at (10.5,0.5) {$\f_1$};
	\end{tikzpicture}
  \caption{The \TWRtext from Figure~\ref{fig:TWR_ex2_full}, view from above}\label{fig:TWR_ex2_above}
\end{figure}

 Therefore the maximum is clearly attained in the interior $\mathring{\msupp}$.

 The following example will show that the unique maximum $\vek_N^{opt}$ of Theorem~\ref{thm:optf_discm} can indeed be attained on $\partial\msupp$, i.e. the case discussed in Corollary~\ref{cor:uniquenessboundary}. For that we add a third investment system to our last example (\ref{eq:normedreturnmatrix}) with the new returns
 \begin{align*}
    t_{1,3},t_{2,3},t_{3,3}=1\text{ and }t_{4,3},t_{5,3}=-1 \quad \text{(hence $\that_3=1$)}
 \end{align*}
 such that the vectors $\quot$, $i=1,\dots,5$, form the matrix
 \begin{align}\label{eq:returnsex3} 
    \tilde A:=(a_{i,k})_{\substack{i=1,\dots,5\\k=1,\dots,3}}=\begin{pmatrix}-\nicefrac{1}{2}&\nicefrac{2}{5}&1\\ \nicefrac{3}{2}&\nicefrac{8}{5}&1\\ 1&-\nicefrac{2}{5}&1\\ -1&\nicefrac{1}{5}&-1\\ \nicefrac{1}{2}&-1&-1\\ \end{pmatrix}\in\R^{5\times 3}
 \end{align}
 This set of trading systems fulfills Assumption~\ref{discm}(\ref{asm2:discm}) since $\sum\limits_{i=1}^{N=5}t_{i,3}=1>0$.

 Assumption~\ref{discm}(\ref{asm3:discm}) is satisfied as well, because the three columns of $\tilde{A}$ are linearly independent.
 For Assumption~\ref{discm}(\ref{asm1:discm})  we have to show that
 \begin{align}\label{ex:ass1} 
	\forall\,\vek\in\partial B_\varepsilon(0)\cap\Lambda_\varepsilon\,\, \exists\, i_0=i_0(\vek),\text{ with }\sprod{\quot[i_0]^\top}{\vek}<0
 \end{align}
 holds. If not, we would have an investment vector
 \begin{align*}
    \hat\vek=\begin{pmatrix}\hat\f_1,\hat\f_2,\hat f\end{pmatrix}\in\partial B_\varepsilon(0)\cap\Lambda_\varepsilon,
 \end{align*}
 such that (\ref{ex:ass1}) is not true for all rows of the matrix $\tilde A$. In particular if we look at lines 4 and 5
 \begin{align*}
    -\hat\f_1+\frac{1}{5}\hat\f_2-\hat f &\geq 0\\
    \frac{1}{2}\hat\f_1 -\hat\f_2-\hat f &\geq 0,
 \end{align*}
 the sum of both inequalities still has to be true
 \begin{align*}
   -\frac{1}{2}\hat\f_1-\frac{4}{5}\hat\f_2 - 2\hat f \geq 0,
 \end{align*}
 which is a contradiction to $\hat\vek$ being an element of $\partial B_\varepsilon(0)\cap\Lambda_\varepsilon\subset\R^3_{\geq0}$.

 Now we examine the following vector of investments
 \begin{align*}
	\vek^{\ast}=\begin{pmatrix}\f_1^{\ast}\\\f_2^{\ast}\\f^{\ast}\end{pmatrix}:=\begin{pmatrix}\f_1^{\ast}\\\f_2^{\ast}\\0\end{pmatrix}
 \end{align*}
 with $(\f_1^{\ast},\f_2^{\ast})^\top\approx(0.4109,0.3425)^\top$ the unique maximum of the optimization problem of the reduced set of trading systems from the last example 
 (cf. (\ref{eq:optfdiscmex})).

 The first derivative of the \TWRtext in the direction of the third component at $\vek^{\ast}$ is given by
 \begin{align*}
	\frac{\partial}{\partial f}\TWR[5](\vek^{\ast})=\underbrace{\TWR[5](\vek^{\ast})}_{>0}\cdot\sum\limits_{i=1}^{N=5}\frac{a_{i,3}}{1+\sprod{\quot^\top}{\vek^{\ast}}}	\approx-0.359<0
 \end{align*}
 Moreover with $\vek^{\ast}$ being the optimal solution of the last example in two variables we have
 \begin{align*}
	\frac{\partial}{\partial \f_1}\TWR[5](\f_1^{\ast},\f_2^{\ast},0)=0=\frac{\partial}{\partial \f_2}\TWR[5](\f_1^{\ast},\f_2^{\ast},0)
	\intertext{and}
	\nicefrac{\partial^2}{\partial \f_i^2}\TWR[5](\f_1^\ast,\f_2^\ast,0)<0,\quad i=1,2.
 \end{align*}
 Thus $\vek^{\ast}$ is indeed a local maximal point on the boundary of $\msupp$ for $\TWR[5]$ with the three trading systems in (\ref{eq:returnsex3}). 
 Corollary~\ref{cor:uniquenessboundary} yields the uniqueness of this maximal solution for
 \begin{align*}
	\underset{\vek\in\msupp}{\text{maximize}}\quad\TWR[5](\vek).
 \end{align*}


\vspace*{0.5cm}
      \section{Conclusion} \label{sec:5}    

 With our main theorems, Theorem~\ref{thm:optf_discm} and Corollary~\ref{cor:uniquenessboundary}, we were
 able give a complete existence and uniqueness theory for the optimization problem \eqref{prob:discm} of a
 multivariate \TWRtext under reasonable assumptions. Furthermore, due to the convexity of the domain $\msupp$
 (Lemma~\ref{lem:convexity}), the concavity of $\left[\text{TWR} (\cdot)\right]^{1/N}$
 (see Lemma~\ref{lem:mconcavity}) and the uniqueness of the ``optimal $f$'' solution, it is always guaranteed
 that simple numerical methods like steepest ascent will find the maximum.

\vspace*{0.8cm}


\end{document}